\documentclass{article}

\usepackage{amsmath,amssymb,amsthm}
\newtheorem{theorem}{Theorem}
\newtheorem{corollary}[theorem]{Corollary}
\newtheorem{definition}[theorem]{Definition}
\newtheorem{lemma}[theorem]{Lemma}
\newtheorem{proposition}[theorem]{Proposition}

\begin{document}
\title{Local Operations and Completely Positive Maps \\ in Algebraic Quantum Field Theory}
\author{Yuichiro Kitajima}
\maketitle

\begin{abstract}
Einstein introduced the locality principle which states that all physical effect in some finite space-time region does not influence its space-like separated finite region. Recently, in algebraic quantum field theory, R\'{e}dei captured the idea of the locality principle by the notion of operational separability. 
The operation in operational separability is performed in some finite space-time region, and leaves unchanged the state in its space-like separated finite space-time region. This operation is defined with a completely positive map. In the present paper, we justify using a completely positive map as a local operation in algebraic quantum field theory, and show that this local operation can be approximately written with Kraus operators under the funnel property.
\end{abstract}

\section{Introduction}

Einstein \cite{einstein1948quanten} introduced the separability principle and the locality principle to show incompleteness of quantum mechanics. The separability principle says that `any two spatially separated systems possess their own separate real states' \cite[p.173]{howard1985einstein}. Einstein writes:

\begin{quotation}
[I]t is characteristic of these physical things that they are conceived of as being arranged in a space-time continuum. Further, it appears to be essential for this arrangement of the things introduced in physics that, at a specific time, these things claim an existence independent of one another, insofar as the these things `lie in different parts of space'. (\cite[p.321]{einstein1948quanten}; Howard's translation \cite[p.187]{howard1985einstein})
\end{quotation}


Einstein introduced the locality principle in addition to the separability principle. Einstein writes:

\begin{quotation}
For the relative independence of spatially distant things ($A$ and $B$), this idea is characteristic: an external influence on $A$ has no immediate effect on $B$; this is known as the `principle of local action', which is applied consistently only in field theory. The complete suspension of this basic principle would make impossible the idea of the existence of (quasi-) closed systems and, thereby, the establishment of empirically testable laws in the sense familiar to us. (\cite[p.322]{einstein1948quanten}; Howard's translation \cite[p.188]{howard1985einstein})
\end{quotation}

This principle states that any physical effect in some finite space-time region does not influence its space-like separated finite region. Einstein \cite{einstein1948quanten} argued for the incompleteness of quantum mechanics under the locality principle and the separability principle.

According to Howard \cite{howard1985einstein}, the Bell inequality is a consequence of the separability and locality principle. Since the Bell inequality does not hold in algebraic quantum field theory and in quantum mechanics \cite{halvorson2000generic, kitajima2013epr, landau1987violation, summers1987bell, summers1987abell, summers1987maximal, summers1988maximal}, we must give up either separability or locality. Howard \cite{howard1985einstein} argued that the separability principle must be abandoned, and that the locality principle holds in quantum theory.
In the present paper we concentrate on the locality principle because it can be compatible with the violation of Bell inequalities. 

Recently, in algebraic quantum field theory, R\'{e}dei \cite{redei2010einstein, redei2011einstein} captured the idea of the locality principle by the notion of operational separability (Definition \ref{separability}), which had been introduced by R\'{e}dei and Valente \cite{redei2010local}. The reason why he adopts the formalism of algebraic quantum field theory is that Einstein \cite{einstein1948quanten} says that physical things are conceived of as being arranged in a space-time continuum, and that observables in algebraic quantum field theory are `explicitly regarded as localized in regions of the space-time continuum' \cite[p.1045]{redei2010einstein}.

The operation in operational separability is performed in some finite space-time region, and leaves unchanged the state in its space-like separated finite region. It is defined with a completely positive map. Valente \cite{valente2013local} called such an operation a relatively local operation (Definition \ref{two-local-operations}). On the other hand, there is another local operation. It is called an absolutely local operation, which is written with some operators in a local algebra which is associated with some open bounded region (Definition \ref{two-local-operations}). This operation in some finite space-time region has no effects on the entire causal complement of this region. A difference between these two types of operations is that a relatively local operation is not necessarily written in terms of local operators while an absolutely local operation is given by local operators by definition. Valente \cite{valente2013local} argued that the concept of absolutely local operation is too strong to express Einstein's locality principle because this principle simply demands that an operation performed in a system $A$ leaves unchanged the state of another space-like separated system $B$. 

There are two tasks here. One is to justify using a completely positive map as a local operation in algebraic quantum field theory. Another is to clarify the relation between these local operations. In the present paper, we show that a local operation in algebraic quantum field theory should be a completely positive map, and that a relatively local operation can be approximately written with some operators as well as an absolutely local operation.

The structure of the paper is as follows. We begin in Section \ref{AQFT} by reviewing the formalism of algebraic quantum field theory and notions of independence. In Section \ref{completely-positive-maps} we examine a definition of an operation. Usually  a completely positive map is regarded as an operation. Although this assumption is natural in the case of nonrelativistic quantum mechanics, it is not transparent in the case of algebraic quantum field theory. We will justify using a completely positive map as a local operation in the case of algebraic quantum field theory (Theorem \ref{completely}). We conclude, in Section \ref{relatively-local-operations}, by examining a similarity between an absolutely local operation and a relatively local operation. An absolutely local operation is written with some operators. This representation is called the Kraus representation. On the other hand, a relatively local operation does not necessarily admit such a representation. By establishing a slightly generalized Kraus representation theorem (Theorem \ref{Kraus1}),  it is shown that a relatively local operation can be approximately written with Kraus operators under the funnel property (Corollary \ref{Kraus2}). 

\section{Algebraic quantum field theory}
\label{AQFT}

Algebraic quantum field theory exists in two versions: the Haag-Araki theory which uses von Neumann algebras on a Hilbert space, and the Haag-Kastler theory which uses abstract C*-algebras. Here we adopt the Haag-Araki theory. In this theory, each bounded open region $\mathcal{O}$ in the Minkowski space is associated with a von Neumann algebra $\mathfrak{N}(\mathcal{O})$ on a Hilbert space $\mathcal{H}$. Such a von Neumann algebra is called a local algebra. 

In the present paper we use the following notation. For a subspace $\mathcal{K}$  of a Hilbert space $\mathcal{H}$, $\{ \mathcal{K} \}^-$ stands for the closure of $\mathcal{K}$. $\mathbb{B}(\mathcal{H})$ is the set of all bounded operators on a Hilbert space $\mathcal{H}$. $I$ stands for an identity operator on a Hilbert space. For a von Neumann algebra $\mathfrak{N}$ on a Hilbert space $\mathcal{H}$, $\mathfrak{N}'$ stands for the commutant of $\mathfrak{N}$ in $\mathbb{B}(\mathcal{H})$. For von Neumann algebras $\mathfrak{N}_1$ and $\mathfrak{N}_2$ on a Hilbert space $\mathcal{H}$, $\mathfrak{N}_1 \vee \mathfrak{N}_2$ stands for the von Neumann algebra generated by $\mathfrak{N}_1$ and $\mathfrak{N}_2$. 

For an open bounded region $\mathcal{O}$ in the Minkowski space, $\mathcal{O}'$ stands for the causal complement of $\mathcal{O}$ and $\bar{\mathcal{O}}$ the closure of $\mathcal{O}$. A double cone in Minkowski space is the intersection of the causal future of a point $x$ with the causal past of a point $y$ to the future of $x$. Two double cones $\mathcal{O}_1$, $\mathcal{O}_2$ are said to be strictly space-like separated if there is a neighborhood $\mathcal{N}$ of zero such that $\mathcal{O}_1 + x$ is space-like separated from $\mathcal{O}_2$ for all $x \in \mathcal{N}$.

In the present paper, we assume the following axioms.
\begin{definition}[Microcausality]
\label{microcausality}
\cite[p.10]{baumgartel1995operatoralgebraic}
Let $\mathcal{O}_1$ and $\mathcal{O}_2$ be bounded open regions in the Minkowski space. If $\mathcal{O}_1 \subseteq \mathcal{O}_2'$, then $\mathfrak{N}(\mathcal{O}_1) \subseteq \mathfrak{N}(\mathcal{O}_2)'$. This property is called microcausality. 
\end{definition}



\begin{definition}[The funnel property]
\label{funnel}
\cite[Definition 6.14]{summers1990independence}
For any pair $(\mathcal{O}, \tilde{\mathcal{O}})$ of double cones in the Minkowski space such that the closure of $\bar{\mathcal{O}} \subset \tilde{\mathcal{O}}$, there exists a type I factor $\mathfrak{N}$ such that $\mathfrak{N}(\mathcal{O}) \subset \mathfrak{N} \subset \mathfrak{N}(\tilde{\mathcal{O}})$. This property is called the funnel property.
\end{definition}

The following property is derived from usual axioms of algebraic quantum field theory \cite[Corollary 1.5.6]{baumgartel1995operatoralgebraic}.

\begin{definition}
\label{properly-infinite}
Let $\mathcal{O}$ be a bounded open region in the Minkowski space. $\mathfrak{N}(\mathcal{O})$ is properly infinite.
\end{definition}

Although there are some different notions of independence \cite{hamhalter2013quantum, summers1990independence}, we use only two notions.

\begin{definition}
Let $\mathfrak{N}_1$ and $\mathfrak{N}_2$ be von Neumann algebras on a Hilbert space $\mathcal{H}$.
\begin{itemize}
\item $\mathfrak{N}_{1}$ and $\mathfrak{N}_{2}$ are called Schlieder independent if $A_{1}A_{2} \neq 0$ whenever $0 \neq A_{1} \in \mathfrak{N}_{1}$ and $0 \neq A_{2} \in \mathfrak{N}_{2}$. 
\item $\mathfrak{N}_{1}$ and $\mathfrak{N}_{2}$ are called split if there exists a type I factor $\mathfrak{N}$ such that $\mathfrak{N}_{1} \subset \mathfrak{N} \subset \mathfrak{N}_{2}'$.
\end{itemize}
\end{definition}

If two double cones $\mathcal{O}_{1}$ and $\mathcal{O}_{2}$ are strictly space-like separated, then $\mathfrak{N}(\mathcal{O}_{1})$ and $\mathfrak{N}(\mathcal{O}_{2})$ are split by Axioms \ref{microcausality} and \ref{funnel}. The following lemma shows that the split property is stronger than the Schlieder property. 

\begin{lemma}
\label{murray}
\cite[Theorem 5.5.4]{kadison1983fundamentals}
Let $\mathfrak{N}$ be a factor on a Hilbert space $\mathcal{H}$. Then $AA' \neq 0$ for any nonzero operators $A \in \mathfrak{N}$ and $A' \in \mathfrak{N}'$.
\end{lemma}





Lemma \ref{murray} shows that von Neumann algebras $\mathfrak{N}_{1}$ and $\mathfrak{N}_{2}$ are Schlieder independent if they are split. The following proposition is a characterization of Schlieder independence.

\begin{proposition} 
\label{Schlieder-proposition}
\cite[Theorem 1 and Proposition 2]{florig1997statistical} \cite[Theorem 11.2.5 and Theorem 11.2.17]{hamhalter2013quantum}
Let $\mathfrak{A}_1$ and $\mathfrak{A}_2$ be mutually commuting C*-subalgebras of a C*-algebra $\mathfrak{A}$. The following conditions are equivalent.
\begin{enumerate}
\item $\mathfrak{A}_1$ and $\mathfrak{A}_2$ are Schlieder independent.
\item $\| A_1 A_2 \| = \| A_1 \| \| A_2 \|$ for any $A_1 \in \mathfrak{A}_1$ and $A_2 \in \mathfrak{A}_2$.
\end{enumerate}
\end{proposition}

\section{Completely positive maps}
\label{completely-positive-maps}
In this section, we examine the reason why local operations are assumed to be completely positive in algebraic quantum field theory.

\begin{definition}
Let $\mathfrak{N}$ be a von Neumann algebra and let $T$ be a linear map of $\mathfrak{N}$.
\begin{itemize}
\item $T$ is called positive if $A \geq 0$ entails $T(A) \geq 0$.
\item Let $[A_{jk}]$ be $n \times n$-matrix with entries $A_{jk}$ in $\mathfrak{N}$. $T$ is called completely positive if $[A_{jk}] \geq 0$ entails $[T(A_{jk})] \geq 0$ for any $n \in \mathbb{N}$.
\end{itemize}
\end{definition}

It is natural to assume that an operation is a positive map because the probability after the process represented by the map $T$ must be positive. Moreover, if we introduce an environmental system which is represented by a set $M_{n}(\mathbb{C})$ of all $n \times n$ matrices with complex entries, then $(T \otimes \text{Id})(A)$ must be also positive for any positive operator $A$ on $\mathbb{B}(\mathcal{H}) \otimes M_{n}(\mathbb{C})$, where $\text{Id}$ denotes the identity map on $M_{n}(\mathbb{C})$. This is equivalent to the condition that $T$ is completely positive. Therefore it is reasonable to assume that an operation is completely positive in the case of nonrelativistic quantum mechanics.

A completely positive map plays an important role in quantum measurements \cite{okamura2016measurement, ozawa1984quantum}. It is also used as a local operation in algebraic quantum field theory \cite{ojima2016local, redei2010einstein, redei2010operational, redei2010quantum, redei2010local, valente2013local}. For example, a new concept of local states is defined in terms of  a completely positive map \cite{ojima2016local}. But it is not transparent to use a completely positive map as an operation in algebraic quantum field theory because any local algebra which is associated with two space-like separated regions is not isomorphic to $\mathbb{B}(\mathcal{H}) \otimes M_{n}(\mathbb{C})$. Therefore, we examine how we can justify it in algebraic quantum field theory in Theorem \ref{completely}. 

We introduce a positive map $T$ of $\mathfrak{N}_{1}$ such that it has an extension to $\mathfrak{N}_{1} \vee \mathfrak{N}_{2}$ which is the identity map on $\mathfrak{N}_2$ to capture an idea that this operation is performed in the system $\mathfrak{N}_{1}$ and it does not influence the system $\mathfrak{N}_{2}$. To examine such an operation, we use the following lemma.

\begin{lemma}
\label{Werner-lemma}
\cite[Lemma]{werner1987local} \cite[Lemma 3.12]{summers1990independence}
Let $\mathfrak{N}_{1}$ and $\mathfrak{N}_{2}$ be mutually commuting von Neumann algebras on a Hilbert space $\mathcal{H}$, and let $T'$ be a positive map of $\mathfrak{N}_{1} \vee \mathfrak{N}_{2}$ such that $T'(A_{2})=A_{2}$ for all $A_{2} \in \mathfrak{N}_{2}$. Then $T'(A_{1}A_{2})=T'(A_{1})A_{2}$ for any $A_{1} \in \mathfrak{N}_{1}$ and $A_{2} \in \mathfrak{N}_{2}$.
\end{lemma}

By using this lemma, we can show the following fact.


\begin{theorem}
\label{completely}
Let $\mathfrak{N}_1$ and $\mathfrak{N}_2$ be mutually commuting von Neumann algebras which are Schlieder independent, let $\mathfrak{N}_2$ have either type $II_1$ direct summand or properly infinite one, and let $T$ be a positive map of $\mathfrak{N}_{1}$. If there is a positive map $T'$ of $\mathfrak{N}_{1} \vee \mathfrak{N}_{2}$ such that 
\[ T'(A_1)=T(A_1), \ \ \ \ \ \ T'(A_2)=A_2 \]
for any $A_1 \in \mathfrak{N}_1$ and $A_2 \in \mathfrak{N}_2$, then $T$ is completely positive.
\end{theorem}

\begin{proof}
Since $\mathfrak{N}_2$ has have either type $II_1$ direct summand or properly infinite one, for any natural number $n \in \mathbb{N}$, there is a set $\{ E_1, \dots, E_n \}$ of mutually orthogonal and equivalent projections in $\mathfrak{N}_2$ \cite[Proposition V.1.35 and Proposition V.1.36]{takesaki2002theory}. Thus there is a set $\{ V_1, \dots, V_n \}$ of partial isometries in $\mathfrak{N}_2$ such that $V_i^*V_i=E_1$ and $V_iV_i^*=E_i$ for any $i \in \{ 1,\dots, n \}$. 


Let $E_{jk} := V_jV_k^*$, let $M_n(\mathfrak{N}_{1})$ be the set of all $n \times n$-matrices $[A_{jk}]$ with entries $A_{jk}$ in $\mathfrak{N}_{1}$, and let
\[ \mathfrak{C} := \Bigg\{ \sum_{j,k=1}^n C_{jk}E_{jk} \Big| C_{jk} \in \mathfrak{N}_1, 1 \leq j,k \leq n \Bigg\}. \]
$\mathfrak{C}$ is a linear subspace of $\mathfrak{N}_1 \vee \mathfrak{N}_2$, and is self-adjoint because $(C_{jk}E_{jk})^*=E_{jk}^*C_{jk}^*=C_{jk}^*E_{kj} \in \mathfrak{C}$ for any $C_{jk} \in \mathfrak{N}_1$. Furthermore, if $C_{jk}, C_{lm} \in \mathfrak{N}_1$, then $(C_{jk}E_{jk})(C_{lm}E_{lm})=\delta_{kl}C_{jk}C_{lm}E_{jm} \in \mathfrak{C}$, where $\delta_{kl}$ equals $1$ if $k=l$, and $0$ if $k \neq l$. By linearity $\mathfrak{C}$ is closed under multiplication. Hence $\mathfrak{C}$ is a *-subalgebra of $\mathfrak{N}_1 \vee \mathfrak{N}_2$.



Let $M_n(\mathfrak{N}_{1})$ be the set of all $n \times n$-matrices $[A_{ij}]$ with entries $A_{ij}$ in $\mathfrak{N}_1$, and let $\alpha$ be a map of $M_n(\mathfrak{N}_{1})$ to $\mathfrak{C}$ such that 
\begin{equation}
\alpha \left( [A_{jk}]
\right)
:=\sum_{j,k=1}^{n} A_{jk} E_{jk} 
\end{equation}
for any $[A_{jk}] \in M_n(\mathfrak{N}_1)$. Clearly $\alpha$ is surjective. Given $S^{(s)}$ and $S^{(t)}$ in $\mathfrak{C}$, say
\[ S^{(s)}=\sum_{j,k=1}^{n} A_{jk}^{(s)}E_{jk}, \ \ \ \ \ \ S^{(t)}=\sum_{j,k=1}^{n} A_{jk}^{(t)}E_{jk}, \]
we have
\begin{equation}
\label{cp1}
\alpha([A_{jk}^{(s)}]^{*})=\alpha([A_{jk}^{(s)}])^{*}, 
\end{equation}
\begin{equation}
\label{cp2}
\alpha([A_{jk}^{(s)}][A_{lm}^{(t)}])=\alpha([A_{jk}^{(s)}])\alpha([A_{lm}^{(t)}]), 
\end{equation}
\begin{equation}
\label{cp3}
\begin{split}
\| A_{jk}^{(s)} - A_{jk}^{(t)} \|&=\| A_{jk}^{(s)} - A_{jk}^{(t)} \| \| E_{jk} \| \\
&=\| (A_{jk}^{(s)} - A_{jk}^{(t)}) E_{jk} \| \\
&=\| E_{jj}(S^{(s)}-S^{(t)})E_{kk} \| \\
&\leq \| S^{(s)}-S^{(t)} \|
\end{split} 
\end{equation}
because $\| E_{jk} \|^2=\| E_{jk}^*E_{jk} \| = \| V_kV_j^*V_jV_k^* \| = \| E_k \| =1$ and $\| A_{jk}^{(s)} - A_{jk}^{(t)} \| \| E_{jk} \|=\| (A_{jk}^{(s)} - A_{jk}^{(t)}) E_{jk} \| $ by Proposition \ref{Schlieder-proposition}. Thus $\alpha$ is a faithful *-homomorphism of $M_n(\mathfrak{N}_{1})$ to $\mathfrak{C}$, which entails that $\mathfrak{C}$ is a C*-algebra \cite[p.192]{takesaki2002theory}.


Let $T$ be a positive map of $\mathfrak{N}_1$, let $T'$ be a positive map of $\mathfrak{N}_{1} \vee \mathfrak{N}_{2}$ such that $T'(A_1)=T(A_1)$ and $T'(A_2)=A_2$ for any $A_1 \in \mathfrak{N}_1$ and $A_2 \in \mathfrak{N}_2$, and let $[A_{jk}]$ be a positive operator in $M_n(\mathfrak{N}_1)$. Then there is $[B_{jk}] \in M_n(\mathfrak{N}_1)$ such that $[A_{jk}]=[B_{jk}]^*[B_{jk}]$, so that $\sum_{j,k=1}^n A_{jk}E_{jk}=\alpha \left( [A_{jk}] \right)=\alpha \left( [B_{jk}]^*[B_{jk}] \right) = \alpha \left( [B_{jk}] \right)^*\alpha \left( [B_{jk}] \right) \geq 0$ by Equations (\ref{cp1}) and (\ref{cp2}).
Since $T'$ is positive on $\mathfrak{N}_1 \vee \mathfrak{N}_2$, $T'(\sum_{j,k=1}^n A_{jk}E_{jk}) \geq 0$. By Lemma \ref{Werner-lemma}, 
\begin{equation}
\sum_{j,k=1}^n T(A_{jk})E_{jk}=\sum_{j,k=1}^n T'(A_{jk})E_{jk}=\sum_{j,k=1}^nT'(A_{jk}E_{jk})=T'\left(\sum_{j,k=1}^n A_{jk}E_{jk} \right) \geq 0. 
\end{equation}

Since $\mathfrak{C}$ is a C*-algebra, there is an operator $D \in \mathfrak{C}$ such that $\sum_{j,k=1}^n T(A_{jk})E_{jk} = D^*D$ \cite[Theorem 4.2.6]{kadison1983fundamentals}. Therefore
\begin{equation}
[T(A_{jk})]=\alpha^{-1}  \left( \sum_{j,k=1}^n T(A_{jk})E_{jk} \right) =\alpha^{-1}(D^*D)=\alpha^{-1}(D)^*\alpha^{-1}(D) \geq 0. 
\end{equation}
Because $n$ is an arbitrary natural number, $T$ is completely positive on $\mathfrak{N}_1$.
\qed
\end{proof}

Let $\mathcal{O}_1$ and $\mathcal{O}_2$ be double cones such that $\mathcal{O}_1 \subset \mathcal{O}_2'$ and let $T$ be a positive map of $\mathfrak{N}(\mathcal{O}_{1})$. When $T$ has an extension to $\mathfrak{N}(\mathcal{O}_{1}) \vee \mathfrak{N}(\mathcal{O}_{2})$ which is the identity map on $\mathfrak{N}(\mathcal{O}_{2})$, $T$ can be regarded as an operation performed in $\mathcal{O}_{1}$ which does not influence a state in $\mathcal{O}_{2}$. Since $\mathfrak{N}(\mathcal{O}_{1})$ and $\mathfrak{N}(\mathcal{O}_{2})$ are split by Definitions \ref{microcausality} and \ref{funnel}, they are Schlieder independent by Lemma \ref{murray}. By Definition \ref{properly-infinite}, any local algebra is properly infinite. Thus, Theorem \ref{completely} entails that $T$ is completely positive. Therefore it is reasonable to assume that a local operation performed in some region which does not influence its space-like separated region is completely positive in algebraic quantum field theory.

\section{Relatively local operations}
\label{relatively-local-operations}

R\'{e}dei and Valente \cite{redei2010local} introduced the notion of operational W*-separability to capture the idea that a causally well behaved operation exists.

\begin{definition}[Operational W*-separability]
\label{separability}
\cite[Definition 6]{redei2010local}
Let $\mathfrak{N}_{1}$ and $\mathfrak{N}_{2}$ be von Neumann subalgebras of a von Neumann algebra $\mathfrak{N}$. $\mathfrak{N}_{1}$ and $\mathfrak{N}_{2}$ are called operationally W*-separable in $\mathfrak{N}$ if the following two conditions are true:
\begin{enumerate}
\item If $T$ is a normal completely positive map of $\mathfrak{N}$ such that $T(A_{1}) \in \mathfrak{N}_{1}$ for any $A_{1} \in \mathfrak{N}_{1}$, there exists a normal completely positive map $T'$ such that $T'(A_{1})=T(A_{1})$ and $T'(A_{2})=A_{2}$ for any $A_{1} \in \mathfrak{N}_{1}$ and $A_{2} \in \mathfrak{N}_{2}$.
\item If $T$ is a normal completely positive map of $\mathfrak{N}$ such that $T(A_{2}) \in \mathfrak{N}_{2}$ for any $A_{2} \in \mathfrak{N}_{2}$, there exists a normal completely positive map $T'$ such that $T'(A_{2})=T(A_{2})$ and $T'(A_{1})=A_{1}$ for any $A_{2} \in \mathfrak{N}_{2}$ and $A_{1} \in \mathfrak{N}_{1}$.
\end{enumerate}

\end{definition}

The normal completely positive map $T'$ in Definition \ref{separability} is performed in some finite space-time region, and leaves unchanged the state in its space-like separated finite region. Thus, this definition requires that there exists such a causally well behaved operation. The following proposition shows that operational W*-separability holds in algebraic quantum field theory.

\begin{proposition}
\cite[Proposition 2]{redei2010operational}; \cite[Section 5]{redei2010quantum}; \cite[Theorem 5.2]{summers2009subsystems}
Let assume microcausality (Definition \ref{microcausality}) and the funnel property (Definition \ref{funnel}), let $\mathcal{O}_1$ and $\mathcal{O}_2$ be strictly space-like separated double cones. Then $\mathfrak{N}(\mathcal{O}_1)$ and $\mathfrak{N}(\mathcal{O}_2)$ are operationally W*-separable in $\mathfrak{N}(\mathcal{O}_1) \vee\mathfrak{N}(\mathcal{O}_2)$.
\end{proposition}

In this section, we examine the normal completely positive map $T'$ in Definition \ref{separability}. Valente \cite{valente2013local} called it a relatively local operation. There is another local operation. It is called an absolutely local operation. Thus there are two types of local operations.

\begin{definition}
\label{two-local-operations}
\cite[Section 3]{valente2013local}
Let $\mathfrak{N}_{1}$ and $\mathfrak{N}_{2}$ be mutually commuting von Neumann algebras on a Hilbert space $\mathcal{H}$.
\begin{itemize}
\item A normal completely positive map $T$ of $\mathbb{B}(\mathcal{H})$ is called an absolutely local operation in $\mathfrak{N}_1$ if there are operators $K_i$ in $\mathfrak{N}_{1}$ such that 
\[ T(A)=\sum_{j \in J} K_{j}^{*}AK_{j}, \ \ \ \ T(I)=I \]
for any $A \in \mathbb{B}(\mathcal{H})$.
\item A normal completely positive map $T$ of $\mathfrak{N}_1 \vee \mathfrak{N}_2$ is called a relatively local operation in $\mathfrak{N}_{1}$ with respect to $\mathfrak{N}_{2}$ if $T(A_{1}) \in \mathfrak{N}_{1}$ and $T(A_{2})=A_{2}$ for any $A_{1} \in \mathfrak{N}_{1}$ and $A_{2} \in \mathfrak{N}_{2}$.
\end{itemize}
\end{definition}

An absolutely local operation $T$ in Definition \ref{two-local-operations} does not influence the system $\mathfrak{N}_1'$ which includes $\mathfrak{N}_2$ while a relatively local operation in Definition \ref{two-local-operations} does not influence only the system $\mathfrak{N}_2$. In the case of algebraic quantum field theory, an absolutely local operation in some region has no effect on the entire causal complement of this region. Although Clifton and Halvorson \cite{clifton2001entanglement} discussed local disentanglement in terms of absolutely local operations, Valente \cite{valente2013local} argued that an absolutely local operation is too strong because Einstein's locality principle simply demands that an operation performed in a system $A$ leaves unchanged the state of another space-like separated system $B$. 

There are two classical theorems characterizing a completely positive map. One is Stinespring representation theorem, and another Kraus representation theorem.

\begin{theorem}[Stinespring representation theorem]
\label{Stinespring}
\cite{stinespring1955positive}
Let $\mathfrak{A}$ be a unital C*-algebra, let $\mathcal{H}$ be a Hilbert space, and let $T$ be a completely positive map from $\mathfrak{A}$ to $\mathbb{B}(\mathcal{H})$. Then there exists a Hilbert space $\mathcal{K}$, a representation $\pi: \mathfrak{A} \rightarrow \mathbb{B}(\mathcal{K})$, and a bounded operator $W: \mathcal{H} \rightarrow \mathcal{K}$ such that
\[ T(A)=W^{*}\pi(A)W \]
for any $A \in \mathfrak{A}$.
\end{theorem}

Kraus representation theorem follows Stinespring representation theorem. 

\begin{theorem}[Kraus representation theorem]
\label{original-Kraus}
\cite{kraus1983states}
Let $\mathcal{H}$ be a Hilbert space and let $T$ be a normal completely positive map of $\mathbb{B}(\mathcal{H})$ such that $0 < T(I) \leq I$. Then there are bounded operators $K_{j}$ in $\mathbb{B}(\mathcal{H})$ such that
\[ T(A)=\sum_{j \in J}K_{j}^{*}AK_{j}, \ \ \ \ \ \ 0 < \sum_{j \in J}K_{j}^{*}K_{j} \leq I \]
for any $A \in \mathbb{B}(\mathcal{H})$. 
\end{theorem}

The operators $K_{i}$ in Theorem \ref{original-Kraus} are called Kraus operators. If a normal completely positive map is defined on a proper subalgebra of $\mathbb{B}(\mathcal{H})$, it does not necessarily admit a decomposition with Kraus operators. 

Here we examine a normal completely positive map $T$ from a type I factor $\mathfrak{N}$ on a Hilbert space $\mathcal{H}$ to $\mathbb{B}(\mathcal{H})$. Note that if $T(A) \in \mathfrak{N}$ for any $A \in \mathfrak{N}$, we can apply Kraus representation theorem because there is a Hilbert space $\mathcal{K}$ such that $\mathfrak{N}$ is isomorphic to $\mathbb{B}(\mathcal{K})$. However, $T(\mathfrak{N})$ is not necessarily included in $\mathfrak{N}$, so we cannot use the original Kraus representation theorem. Yet, we show below (Theorem \ref{Kraus1}) that a representation theorem similar to Kraus representation theorem holds if the von Neumann algebra $\mathfrak{N}$ is a type I factor.

\begin{theorem}
\label{Kraus1}
Let $\mathfrak{N}$ be a type I factor on a Hilbert space $\mathcal{H}$, and let $T$ be a normal completely positive map of $\mathfrak{N}$ to $\mathbb{B}(\mathcal{H})$ such that $0 < T(I) \leq I$. Then
there are bounded operators $K_j$ in $\mathbb{B}(\mathcal{H})$ such that 
\[ T(A)=\sum_{j \in J} K_j^*AK_j, \ \ \ \ \ \ 0 < \sum_{j \in J} K_{j}^{*}K_{j} \leq I \]
for any $A \in \mathfrak{N}$. 

\end{theorem}

\begin{proof}
By Theorem \ref{Stinespring}, there is a representation $\pi$ of $\mathfrak{N}$ on a Hilbert space $\mathcal{K}$ and a bounded operator $W: \mathcal{H} \rightarrow \mathcal{K}$ such that $T(A)=W^*\pi(A)W$ for any $A \in \mathfrak{N}$. Since $T$ is normal, so is $\pi$. Since $\pi(I)>0$ and $\mathfrak{N}$ is a type I factor, there exists a minimal projection $P_0 \in \mathfrak{N}$ such that $\pi(P_0) \neq 0$. Let $x_0 \in \mathcal{H}$ be a unit vector such that $P_0x_0=x_0$, let $y_0 \in \mathcal{K}$ be a unit vector such that $\pi(P_0)y_0=y_0$, and let $E_0$ and $Q_0$ be projections whose ranges are $\{ \pi(\mathfrak{N})y_0 \}^{-}$ and $\{ \mathfrak{N}x_0 \}^{-}$, respectively. Then $E_0 \in \pi(\mathfrak{N})'$.
For any $A \in \mathfrak{N}$, $P_0AP_0=\langle x_0, A x_0 \rangle P_0$ since $P_0$ is a minimal projection. 
Thus
$\langle y_0, \pi(A)y_0 \rangle = \langle y_0, \pi(P_0AP_0)y_0 \rangle = \langle x_0, Ax_0 \rangle$
for any $A \in \mathfrak{N}$. 
Therefore there exists a unitary operator $U_0$ from $\{ \pi(\mathfrak{N})y_0 \}^{-}$ to $\{ \mathfrak{N}x_0 \}^-$ such that $\pi(A)E_0=U_0^*AU_0$ for any $A \in \mathfrak{N}$ by \cite[Proposition 4.5.3]{kadison1983fundamentals}. 
Let $V_0:= Q_0U_0E_0$. Then $V_0$ is an isometry from $\mathcal{K}$ to $\mathcal{H}$ such that $\pi(A)E_0=V_0^{*}AV_0$ for any $A \in \mathfrak{N}$.

By Zorn's lemma, it can be shown that there are a maximal family $\{ E_j \in \pi(\mathfrak{N})' | j \in J \}$ of mutually orthogonal projections in $\pi(\mathfrak{N})'$ and a family $\{ V_j | j \in J \}$ of isometries from $\mathcal{K}$ to $\mathcal{H}$ such that the range of $E_j$ is $\{ \pi(\mathfrak{N})y_j \}^-$ for some unit vector $y_j \in \mathcal{K}$, and $\pi(A)E_j=V_j^*AV_j$ for any $A \in \mathfrak{N}$. Suppose that $\sum_{j \in J} E_{j} < I$. Let $F^0 := I - \sum_{j \in J} E_{j}$. Then there is a unit vector $y' \in F^0 \mathcal{K}$. Since $\pi(I)y'=y' \neq 0$ and $\mathfrak{N}$ is a type I factor, there is a minimal projection $P^0 \in \mathfrak{N}$ such that $\pi(P^0)y' \neq 0$. Thus $\pi(P^0)F^0 \neq 0$. Let $x^0$ be a unit vector such that $P^0x^0=x^0$, let $y^0$ be a unit vector such that $\pi(P^0)F^0 y^0 = y^0$, and let $E^0$ be a projection whose range is $\{ \pi(\mathfrak{N})y^0 \}^-$. Then $E^0 \in \pi(\mathfrak{N})'$. Since $\pi(P^0)y^0=y^0$ and $F^0y^0=y^0$,
\begin{equation}
\langle \pi(A)y_{j}, \pi(B)y^{0} \rangle =\langle \pi(B^{*}A)y_{j}, y^{0} \rangle = \langle E_{j}\pi(B^{*}A)y_{j}, F^{0}y^{0} \rangle =0, 
\end{equation}
and
\begin{equation}
\langle y^{0}, \pi(A)y^{0} \rangle = \langle y^{0}, \pi(P^0AP^0)y^{0} \rangle = \langle x^0, Ax^0 \rangle 
\end{equation}
for any $j \in J$ and $A, B \in \mathfrak{N}$.
Therefore $E_jE^0=0$ for any $j \in J$, and there exists an isometry $V^{0}$ from $\mathcal{K}$ to $\mathcal{H}$ such that $\pi(A)E^{0}=V^{0*}AV^{0}$ for any $A \in \mathfrak{N}$. This contradicts the maximality of $\{ E_{j} | j \in J \}$. Therefore, $\sum_{j \in J} E_{j} = I$.




Let $K_j := V_jW$ for any $j \in J$. Then
\begin{equation}
T(A) = W^*\pi(A)W=\sum_{j \in J} W^*\pi(A)E_jW=\sum_{j \in J} W^*V_j^*AV_jW=\sum_{j \in J}K_j^*AK_j 
\end{equation}
for any $A \in \mathfrak{N}$. Since $0 < T(I) \leq I$ and $T(I)=\sum_{j \in J}K_j^*K_j$, $0 < \sum_{j \in J}K_j^*K_j \leq I$.
\qed
\end{proof}


Under the funnel property (Definition \ref{funnel}), type I factors exist which are interpolated between local algebras of regions strictly contained in each other. By using Theorem \ref{Kraus1}, we show that a relatively local operation can be approximately written with Kraus operators in algebraic quantum field theory.

\begin{corollary}
\label{Kraus2}
Let's assume microcausality (Definition \ref{microcausality}) and the funnel property (Definition \ref{funnel}), let $\tilde{\mathcal{O}}_1$ and $\tilde{\mathcal{O}}_2$ be double cones such that $\tilde{\mathcal{O}}_1 \subset \tilde{\mathcal{O}}_2'$, and let $T$ be a relatively local operation in $\mathfrak{N}(\tilde{\mathcal{O}}_{1})$ with respect to $\mathfrak{N}(\tilde{\mathcal{O}}_{2})$. For any double cones $\mathcal{O}_{1}$ and $\mathcal{O}_{2}$ such that $\bar{\mathcal{O}}_{1} \subset \tilde{\mathcal{O}}_1$ and $\bar{\mathcal{O}}_{2} \subset \tilde{\mathcal{O}}_2$, there are bounded operators $K_j$ in $\mathfrak{N}(\mathcal{O}_{2})'$ such that
\[ T(A)=\sum_{j \in J} K_j^*AK_j, \ \ \ \ \ \sum_{j \in J} K_{j}^{*}K_{j}=I \]
for any $A \in \mathfrak{N}(\mathcal{O}_1) \vee \mathfrak{N}(\mathcal{O}_{2})$.
\end{corollary}

\begin{proof}
Let $\mathcal{O}_{1}$ and $\mathcal{O}_{2}$ be double cones such that $\bar{\mathcal{O}}_{1} \subset \tilde{\mathcal{O}}_1$ and $\bar{\mathcal{O}}_{2} \subset \tilde{\mathcal{O}}_2$. By Axiom \ref{funnel}, there are type I factors $\mathfrak{N}_{1}$ and $\mathfrak{N}_{2}$ such that $\mathfrak{N}(\mathcal{O}_{1}) \subset \mathfrak{N}_{1} \subset \mathfrak{N}(\tilde{\mathcal{O}}_1)$ and $\mathfrak{N}(\mathcal{O}_{2}) \subset \mathfrak{N}_{2} \subset \mathfrak{N}(\tilde{\mathcal{O}}_2)$. Then $\mathfrak{N}(\mathcal{O}_{1}) \vee \mathfrak{N}(\mathcal{O}_{2}) \subset \mathfrak{N}_{1} \vee \mathfrak{N}_{2} \subset \mathfrak{N}(\tilde{\mathcal{O}}_1) \vee \mathfrak{N}(\tilde{\mathcal{O}}_2)$, and $\mathfrak{N}_{1} \vee \mathfrak{N}_{2}$ is a type I factor. By Theorem \ref{Kraus1}, there exists a set $\{K_j | j \in J \}$ of operators in $\mathbb{B}(\mathcal{H})$ such that 
\[ T(A)=\sum_{j \in J} K_j^*AK_j \]
for any $A \in \mathfrak{N}_{1} \vee \mathfrak{N}_{2}$. $T(I)=I$ entails $\sum_{j \in J} K_{j}^{*}K_{j}=I$.

Since $T(A_2)=A_2$ for any $A_2 \in \mathfrak{N}(\mathcal{O}_2)$ and $T(I)=I$, $\sum_{j \in J}[K_j,A_2]^*[K_j,A_2]=0$ \cite[p.13]{clifton2001entanglement}. Thus $K_j \in \mathfrak{N}(\mathcal{O}_2)'$ for any $j \in J$.
\qed
\end{proof}





In Corollary \ref{Kraus2}, double cones $\mathcal{O}_1$ and $\mathcal{O}_2$ can approximate $\tilde{\mathcal{O}_1}$ and $\tilde{\mathcal{O}_2}$, respectively, as closely as possible. So we can say that $T$ can be approximately written with operators in $\mathfrak{N}(\mathcal{O}_2)'$.




\section{Conclusion}
Einstein \cite{einstein1948quanten} introduced the locality principle which states that physical effects in some finite space-time region do not influence its space-like separated finite region. In algebraic quantum field theory, R\'{e}dei \cite{redei2010einstein} captured the idea of the locality principle by the notion of operational W*-separability (Definition \ref{separability}), which had been introduced by R\'{e}dei and Valente \cite{redei2010local}. Valente \cite{valente2013local} called such an operation a relatively local operation to distinguish it from an absolutely local operation which can be written with Kraus operators (Definition \ref{two-local-operations}).

In the present paper, we examined two questions;

\begin{itemize}
\item Can we justify using a completely positive map as a local operation in algebraic quantum field theory?
\item Can we write a relatively local operation with some operators?
\end{itemize}

Roughly speaking, complete positiveness of an operation $T$ in a system $A$ is equivalent to the condition that $T$ performed in the system $A$ does not influence a space-like separated system $B$ which is represented by a set $M_n(\mathbb{C})$ of all $n \times n$ matrices with complex entries in the case of nonrelativistic quantum mechanics. But it is not obvious why a completely positive map is used as an operation in the case of algebraic quantum field theory because any local algebra which is associated with two space-like separated regions is not isomorphic to $\mathbb{B}(\mathcal{H}) \otimes M_{n}(\mathbb{C})$. In Theorem \ref{completely}, we showed that an operation is completely positive in algebraic quantum field theory if it is  performed in some region and does not influence its space-like separated region. Thus, it is reasonable to assume that a local operation is completely positive.

Valente \cite{valente2013local} distinguished between absolutely local operations and relatively local operations. 
A difference between these operations is that a relatively local operation is not necessarily written with Kraus operators while an absolutely local operation is written with Kraus operators by definition (Definition \ref{two-local-operations}). In the present paper, by generalizing slightly Kraus representation theorem (Theorem \ref{Kraus1}), it was shown that a relatively local operation can be approximately written with Kraus operators under the funnel property (Corollary \ref{Kraus2}). 

\section*{acknowledgement}
The author wishes to thank Masanao Ozawa for helpful comments on an earlier draft. The author is supported by the JSPS KAKENHI No.15K01123 and No.23701009.

\bibliographystyle{plain}
\bibliography{kitajima}

\end{document}